\documentclass[11pt]{article}

\usepackage{natbib}

\usepackage[top=1in, bottom=1in, left=1in, right=1in]{geometry}
\usepackage[nolists]{endfloat}
\usepackage{setspace}

\usepackage{amsmath}
\usepackage{amsthm}
\usepackage{amssymb}
\usepackage{graphicx}

\usepackage{mathrsfs} 

\title{Diversity, disparity, and evolutionary rate estimation \\ for unresolved Yule trees}

\author{Forrest W. Crawford and Marc A. Suchard}

\date{Typeset \today}

\newcommand{\pd}[2]{\frac{\partial #1}{\partial #2}}

\newcommand{\od}[2]{\frac{\mathrm{d} #1}{\mathrm{d} #2}}

\newcommand{\dx}[1]{\ \mathrm{d} #1}
\newcommand{\E}{\mathbb{E}}

\renewcommand{\Re}{\text{Re}}

\newcommand{\argmax}[1]{\underset{#1}{\operatorname{argmax}}\ }

\begin{document}

\maketitle
		
\doublespacing

\begin{abstract}
\noindent The branching structure of biological evolution confers statistical dependencies on phenotypic trait values in related organisms.  For this reason, comparative macroevolutionary studies usually begin with an inferred phylogeny that describes the evolutionary relationships of the organisms of interest.  The probability of the observed trait data can be computed by assuming a model for trait evolution, such as Brownian motion, over the branches of this fixed tree.  However, the phylogenetic tree itself contributes statistical uncertainty to estimates of other evolutionary quantities, and many comparative evolutionary biologists regard the tree as a nuisance parameter.  In this paper, we present a framework for analytically integrating over unknown phylogenetic trees in comparative evolutionary studies by assuming that the tree arises from a continuous-time Markov branching model called the Yule process.  To do this, we derive a closed-form expression for the distribution of phylogenetic diversity, which is the sum of branch lengths connecting a set of taxa.  We then present a generalization of phylogenetic diversity which is equivalent to the expected trait disparity in a set of taxa whose evolutionary relationships are generated by a Yule process and whose traits evolve by Brownian motion.  We derive expressions for the distribution of expected trait disparity under a Yule tree.  Given one or more observations of trait disparity in a clade, we perform fast likelihood-based estimation of the Brownian variance for unresolved clades.  Our method does not require simulation or a fixed phylogenetic tree.  We conclude with a brief example illustrating Brownian rate estimation for thirteen families in the Mammalian order Carnivora, in which the phylogenetic tree for each family is unresolved. \\[1em]

\end{abstract}


\section{Introduction}

Evolutionary relationships between organisms induce statistical dependencies in their phenotypic traits \citep{Felsenstein1985Phylogenies}.  Closely related species that have been evolving separately for only a short time will generally have similar trait values, and species whose most recent common ancestor is more distant will often have dissimilar trait values \citep{Harvey1991Comparative}.  However, the origins of phenotypic diversity are still poorly understood \citep{Eldredge1972Punctuated,Gould1977Punctuated,Ricklefs2006Time,Bokma2010Time}.  Even simple idealized models of evolutionary change can give rise to highly varying phenotype values \citep{Foote1993Contributions,Sidlauskas2007Testing}, and researchers disagree about the relative importance of time, the rate of speciation, and the rate of phenotypic evolution in generating phenotypic diversity \citep{Ricklefs2004Cladogenesis,Purvis2004Evolution,Ricklefs2006Time}.

Comparative phylogenetic studies seek to explain phenotypic differences between groups of taxa, and stochastic models of evolutionary change have assisted in this task.  Researchers often treat phenotypic evolution as a Brownian motion process occurring independently along the branches of a \emph{fixed} macroevolutionary tree \citep{Felsenstein1985Phylogenies}. In comparative studies, the Brownian motion model of trait evolution has a convenient consequence: given an evolutionary tree topology and branching times, the trait values at the concurrently observed tips of the tree are distributed according to a multivariate normal random variable.  Brownian motion on a fixed phylogenetic tree is the basis for the most popular regression-based methods for comparative inference and hypothesis testing \citep{Grafen1989Phylogenetic,Garland1992Procedures,Martins1997Phylogenies,Blomberg2003Testing,OMeara2006Testing,Revell2010Phylogenetic}.  In the regression approach, inference of evolutionary parameters of interest becomes a two-step process: first, one must infer a phylogenetic tree; then, \emph{conditional on that tree}, one estimates relevant evolutionary parameters, usually by maximizing likelihood of the observed trait data under the model for trait evolution.  Unfortunately, the uncertainty involved in estimating the tree propagates into the comparative analysis in a way that is difficult to account for (but see \citet{Stone2011Why}), and comparative researchers often lack a precise phylogenetic tree on which to base a regression analysis of trait data.  Modern techniques for dealing with this issue generally resort to simulation.  Some researchers simulate a large number of possible trees and estimate parameters conditional on a single representative tree, such as the maximum clade credibility tree (see, for example, \citet{Alfaro2009Nine}).  Alternative approaches that utilize simultaneous simulation of trees and parameters via Bayesian methods are gaining in popularity \citep{Sidlauskas2007Testing,Slater2012Fitting,Drummond2012Bayesian}.

However, simulation methods can be extremely slow and may require assumptions about prior distributions of unknown parameters that are difficult to justify.  Indeed, in macroevolutionary studies, the phylogenetic tree is often not of interest \emph{per se}, but must be taken into account in order to accurately model the dependency of the traits under consideration.  Many comparative phylogeneticists regard the evolutionary tree as a nuisance parameter in the larger evolutionary statistical model.  For this reason, there is increased interest in tree-free methods of comparative analysis that preserve information about the variance of phenotypic values within unresolved clades \citep{Bokma2010Time}.  

To develop a method for comparative inference in evolutionary studies that does not rely on a particular tree, it is convenient to specify a \emph{generative} model for phylogenetic trees.  In the Yule (pure-birth) process, every existing species independently gives birth with instantaneous rate $\lambda$; when there are $n$ species, the total rate of speciation is $n\lambda$ \citep{Yule1925Mathematical}.  The Yule process is widely used as a null model in evolutionary hypothesis testing and can provide a plausible prior distribution on the space of evolutionary trees in Bayesian phylogenetic inference \citep{Nee1994Reconstructed,Rannala1996Probability,Nee2006Birth}.  One can easily derive finite-time transition probabilities \citep{Bailey1964Elements}, and efficient methods exist to simulate samples from the distribution of Yule trees, conditional on tree age, number of species, or both \citep{Stadler2011Simulating}.  Interestingly, some researchers have pointed out that even the simple Yule process can have unexpected properties that may be relevant in evolutionary theory and reconstruction \citep{Gernhard2008Stochastic,Steel2010Expected}.  

Due to Yule trees' simple Markov branching structure and analytically tractable transition probabilities, many researchers have made progress in characterizing summary properties of the Yule process -- that is, integrating over all Yule tree realizations.  For example, \citet{Steel2002Shape} study aspects of the shape of phylogenies under the Yule model, such as the distribution of the number of edges separating a subset of the extant taxa from the MRCA; \citet{Gernhard2008Stochastic} find distributions of branch lengths; \citet{Steel2010Expected} study the expected length of pendant and interior edges of Yule trees; and \citet{Steel2001Properties} and \citet{Mulder2011Probability} study the distribution of the number of internal nodes separating taxa.  

One important summary statistic for trees in biodiversity applications is \emph{phylogenetic diversity} (PD), defined as the sum of all branch lengths in the minimum spanning tree connecting a set of taxa \citep{Faith1992Conservation}.  Applied researchers in evolutionary biology have found PD to be useful in conservation and biodiversity applications; see, e.g., \citet{Webb2002Phylogenies}, \citet{Moritz2002Strategies}, and \citet{Turnbaugh2008Core}.  PD also has the virtue of being a mathematically tractable statistic for phylogenetic trees, and has attracted interest from researchers interested in its properties.  For example, \citet{Faller2008Distribution} show that the asymptotic distribution (as the number of taxa $n\to\infty$) of PD is normal and give a recursion for computing the distribution of PD where edge lengths are integral.  \citet{Mooers2011Branch} discuss branch lengths on Yule trees and expected loss of PD in conservation applications.  Most importantly for our study, \citet{Stadler2011Distribution} find the moment-generating function for PD conditional on $n$ extant taxa and tree age $t$ under the Yule model.  Following on these inspiring results, we seek now to study analytic properties of Yule trees that are useful for comparative evolutionary studies.

In this paper, we present a framework for computing probability distributions related to diversity and quantitative trait evolution over unresolved Yule trees and describe methods for estimating related parameters.  We first give a mathematical description of  the Yule model of speciation and briefly discuss its properties.  Next, we introduce the Markov reward process, a probabilistic method for deriving probability distributions related to the accumulation of diversity under a Yule model.  In Theorem 1, we give an expression for the probability distribution of PD under a Yule model, conditional on the number of species $n$, time to the most recent common ancestor (TMRCA) or tree age $t$, and speciation rate $\lambda$.  We then demonstrate an important and previously unappreciated relationship between \emph{trait disparity}, the sample variance for a group of taxa \citep{OMeara2006Testing}, and PD for traits evolving on a Yule tree via Brownian motion.  Theorem 2 gives an expression for the distribution of expected trait disparity when integrating over the branch lengths of a Yule tree.  Next, we describe a statistical method for performing fast maximum likelihood estimation of Brownian variance, given an unresolved clade and observed trait disparity.  Our approach does not require fixing a phylogenetic tree or specification of prior probabilities for unknown parameters.  The method is simulation-free and does not seek to infer branch lengths or ancestral states.  We show empirically that our estimators are asymptotically consistent.  We conclude with an application of our method to body size evolution in the Mammalian order Carnivora.


\section{Mathematical background}

To aid in exposition, we briefly establish some notation. Denote the \emph{topology} of a phylogenetic tree by $\tau$.  A topology is the shape of a tree, disregarding branch lengths or age.  We always condition our calculations on the phylogenetic tree having age $t$ with $n$ extant taxa.  Let $t_1,t_2,\ldots$ denote the branching points of a tree, where $t_k$ is the time of branching from $k$ to $k+1$ lineages.  We measure time in the forward direction, so at the TMRCA, $t=0$.  This is in keeping with our mechanistic orientation: the Yule process, to be developed below, runs forward in time from 0 to $t$.

\subsection{Yule processes}

Let $Y(t) \in \{1,2,\ldots\}$ be a Yule process with birth rate $\lambda$ that keeps track of the number of species at time $t$.  The transition probabilities $P_{mn}(t) = \Pr(Y(t)=n\mid Y(0)=m)$ satisfy the Kolmogorov forward equations
\begin{equation}
\begin{split}
 \od{P_{m1}(t)}{t} &= -\lambda P_{m1}(t), \quad\text{and}   \\
 \od{P_{mn}(t)}{t} &= -\lambda n P_{mn}(t) + (n-1)\lambda P_{m,n-1}(t)
\end{split}
\label{eq:yuleodes}
\end{equation}
for $n\geq 1$.  This infinite system of ordinary differential equations can be solved to yield closed forms for the finite-time transition probabilities,
\begin{equation}
P_{mn}(t) = \binom{n-1}{m-1}e^{-m\lambda t} (1-e^{-\lambda t})^{n-m}, 
\label{eq:yuleprob}
\end{equation}
which have a negative binomial form \citep{Bailey1964Elements}.  In the Yule process, we are only concerned with the number of species that exist at any moment in time, not their genealogy.  That is, we assume that the lineage that branches is chosen uniformly from all extant lineages.  The transition probability \eqref{eq:yuleprob} is useful for performing statistical inference: suppose we know the branching rate $\lambda$ and the age $t$ of a tree, and we observe $Y(t)=n$. Then \eqref{eq:yuleprob} gives the \emph{likelihood} of our observation.  Figure \ref{fig:treestep} shows an example realization of a Yule tree, with the corresponding counting process diagram below.  In this example, $\lambda=2$, $Y(0)=2$, and $Y(t=1)=12$.

\begin{figure}
\centering
\includegraphics[width=0.7\textwidth]{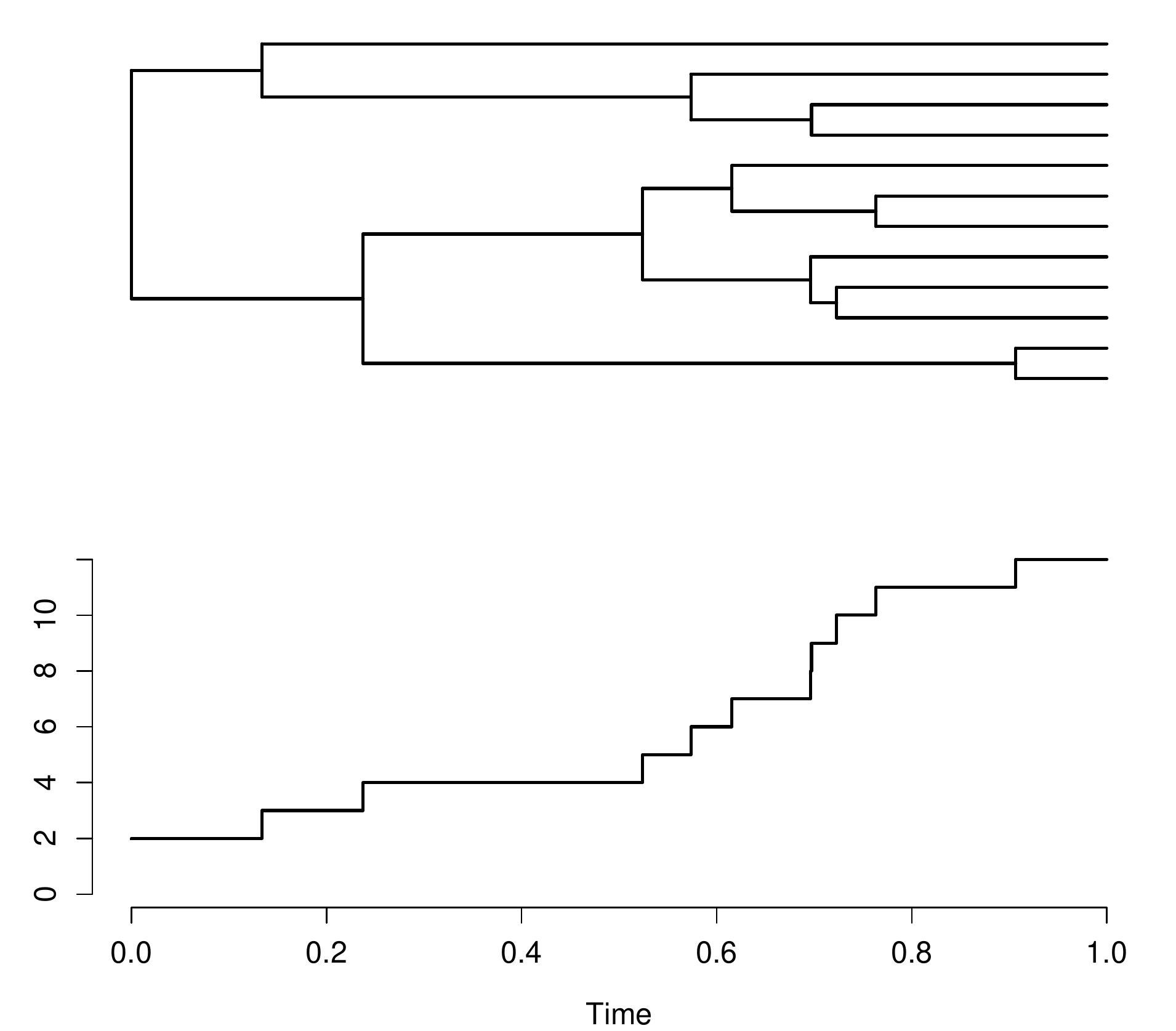}
\caption[Example of a Yule process]{Example of a Yule (pure-birth) tree with the corresponding counting process $Y(t)$ below.  The birth rate in this example is $\lambda=2$.  In the counting process representation of this realization, we only keep track of the total number of species in existence at each time.  }  
\label{fig:treestep}
\end{figure}

%

\subsection{Markov reward processes}

In a Markov reward process, a non-negative reward $a_k$ accrues for each unit of time a Markov process spends in state $k$ \citep{Neuts1995Algorithmic}.  Consider a Yule process $Y(s)$ beginning at $Y(0)=1$ and ending at $Y(t)=n$. The accumulated reward up to time $t$ is 
\begin{equation}
 R_t = \int_0^t a_{Y(s)} \dx{s} .
\label{eq:rewardint}
\end{equation}
When $Y(s)$ is observed continuously from time $0$ to $t$, the process $a_{Y(s)}$ is a fully-observed step function, and $R_t$ can be easily computed as the area under that function.  To illustrate, suppose that the process makes jumps at times $t_1, \ldots, t_{n-1}$, and we define $t_0=0$ and $t_n=t$.  We assume $Y(s)$ is right-continuous, so $Y(t_i) = i+1$.  Then at time $t$, the accumulated reward is
\begin{equation}
 R_t = \sum_{i=1}^n a_{Y(t_{i-1})} (t_i-t_{i-1}) = \sum_{i=1}^n a_i(t_i-t_{i-1}) .
\label{eq:Rtdiff}
\end{equation}
When only $Y(0)$ and $Y(t)$ are observed, it can be challenging to compute the distribution of $R_t$.  In our proofs of Theorems 1 and 2, we appeal to the method developed by \citet{Neuts1995Algorithmic} and \citet{Minin2008Counting} to find reward probabilities conditional on $Y(0)$ and $Y(t)$.  Let 
\begin{equation}
  v_{mn}(x,t) = \Pr(R_t=x,Y(t)=n \mid Y(0)=m)
\end{equation}
be the joint probability that the reward at time $t$ is $x$ and the process is in state $n$, given that the process began in state $m$ at time $0$.  This joint probability formulation is more mathematically convenient than the more natural conditional probability, as we demonstrate in the proofs of the Theorems.  However, it is easy to transform $v_{mn}(x,t)$ into the conditional probability via Bayes' rule, as we show below.  Appendix \ref{app:reward} gives a preliminary lemma deriving a representation for Yule reward processes that will be useful in proving the Theorems that follow.

\section{The distribution of phylogenetic diversity in a Yule process} 

The Yule process is a simple and analytically tractable mechanistic model for producing the birth times of a clade.  If we assume that the species that undergoes speciation is chosen randomly from the extant species at that time, then the Yule process is also a distribution over bifurcating trees of age $t$.  The Markov rewards framework provides a technique to understand integrals over Yule processes in precisely this context.  In this section, we study the distribution of PD in trees generated by a Yule process.  PD depends only on the branching times, and not the underlying topology, of the phylogenetic tree, making it a suitable first step in our goal of integrating over trees in comparative studies.  

To proceed, let $Y(s)$ be a Yule process with branching rate $\lambda$ that keeps track of the number of lineages at time $s$.  We seek an expression for PD, the total branch length of the tree, which is equivalent to the area under the trajectory of the counting process $Y(s)$.  Define a Markov reward process with $Y(s)$ and $a_k=k$ for $k=1,2,\ldots$.  Then 
\begin{equation}
 R_t = \int_0^t a_{Y(s)} \dx{s} = \int_0^t Y(s) \dx{s} .
\end{equation}
We now state our first Theorem giving an expression for the distribution of $R_t$ in a Yule process.
\newtheorem{thm}{Theorem}
\begin{thm}
\label{thm:yule}
For a Yule process with birth rate $\lambda$, starting at $Y(0)=m$ and ending at $Y(t)=n$,
\begin{equation}
   v_{mn}(x,t) = \begin{cases} \delta(x-mt)e^{-m\lambda t} & m=n \\[1em]
  \displaystyle\frac{\lambda^{n-m} e^{-\lambda x}}{(n-m-1)!} \sum_{j=m}^n \binom{n-1}{j-1} \binom{j-1}{m-1} (-1)^{j-m} (x-j t)^{n-m-1}  H(x-j t) & n>m \end{cases}
\label{eq:yulerewardprob}
\end{equation}
where $\delta(x)$ is the Dirac delta function and $H(x)$ is the Heaviside step function.
\end{thm}
\noindent The proof of this Theorem is given in Appendix \ref{app:yule}.  There has been disagreement about whether the definition of PD in different contexts includes the root lineage \citep{Faith1992Conservation,Faith2006Phylogenetic,Crozier2006Conceptual,Faith2006Role}.  We do not take a stance on this issue but note that if $t$ is the stem age of an unresolved clade, then taking $a_1=1$ in \eqref{eq:rewardint} includes the root in the distribution of accumulated PD, and $a_1=0$ does not.  The form of \eqref{eq:yulerewardprob} will change slightly if $m=1$ and $a_1=0$.  

The probability distribution of PD, conditional on $Y(0)=m$ and $Y(t)=n$, is
\begin{equation}
 f_Y(x\mid m,n,t,\lambda) = \frac{v_{mn}(x,t)}{P_{mn}(t)} ,
\label{eq:fYyule}
\end{equation}
where $P_{mn}(t)$ is the Yule transition probability \eqref{eq:yuleprob}.  This family of probability distributions has some interesting properties.  Figure \ref{fig:yulerewards} shows $f_Y(x|m,n,t,\lambda)$ for $m=1$ (with $a_1=1$), $n=1,\ldots,8$, $\lambda=1.2$, and $t=1$.  The unusual shape of the distribution for smaller $n$ demonstrates the piecewise nature of the density, apparent in the functional form \eqref{eq:yulerewardprob}.  Interestingly, \citet{Faller2008Distribution} show that the distribution PD tends toward a normal distribution as $n\to\infty$, a fact suggested by the shape of the distributions in Figure \ref{fig:yulerewards}.  

\begin{figure}
\begin{center}
\includegraphics[width=\textwidth]{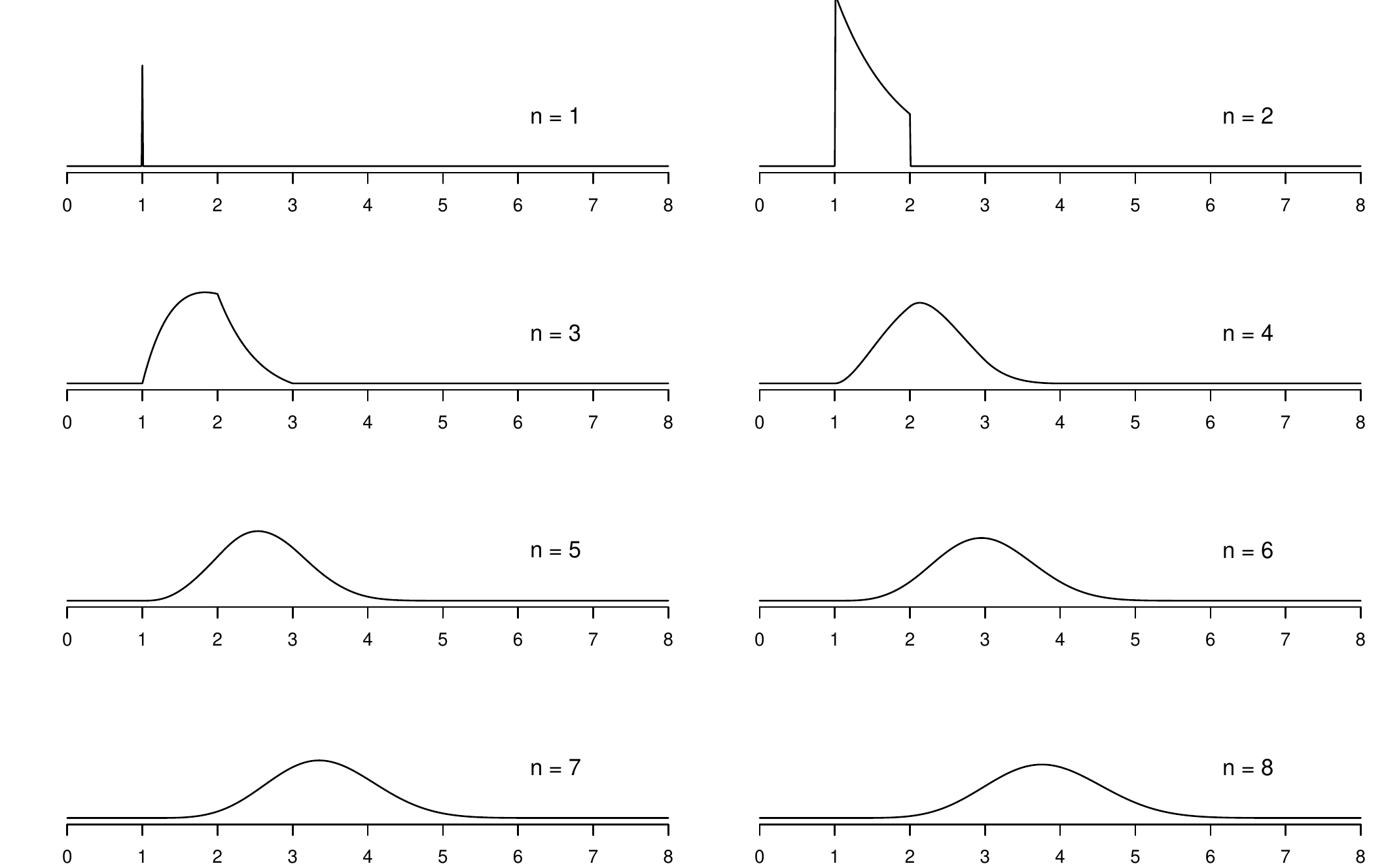}
\end{center}
\caption[Probability density of PD under the Yule process]{Probability densities of phylogenetic diversity (PD), or total branch length, under the Yule process starting at $Y(0)=1$, ending at $Y(t)=n$ for $n=1,\ldots,8$, with $t=1$ and $\lambda=1.2$.  When $n=1$, no births have occurred, so the accrued PD must be exactly $x=t=1$, which we represent here as a point mass at 1.  For $n=2$, the minimum accumulated PD is one, since the process spent at most one unit of time with one species; likewise the maximum accumulated PD is two, since the process spent at most one unit of time with two species. The functional form of \eqref{eq:yulerewardprob} reveals the piecewise nature of the density, which gradually becomes smoother as $n$ becomes large. The vertical probability axis is the same for all plots. }
\label{fig:yulerewards}
\end{figure}

These distributions have some practical uses.  First, one can predict the PD that will arise under the Yule model from a collection of extant species up to time $t$ in the future.  Second, we can calculate the probability that future PD at time $t$ in one group is greater than in the other, conditional on the number of species and diversification rates in both groups; this probability may have uses in conservation applications.  Third, conditional on an inferred phylogenetic tree for a set of $n$ taxa, one could compute the resulting PD $x$ and perform a hypothesis test to evaluate the Yule-PD model using the quantity
\begin{equation}
 \Pr(\text{PD} > x) = \int_x^{nt} f_Y(x) \dx{x} 
\end{equation}
where $f_Y(x)$ is given by \eqref{eq:fYyule} and $nt$ is the maximum PD that can accumulate in time $t$, conditional on $Y(t)=n$.

\section{The distribution of expected phenotypic variance}

\label{sec:disparity}

Since researchers generally do not know the phylogenetic tree for a set of species with certainty, PD is not observable until after a tree has been estimated.  Unfortunately, the uncertainty involved in estimating a tree propagates into subsequent estimates of PD based on that tree, and our distributional results may no longer apply.  We therefore seek a distribution for an analogous quantity that is observable directly from knowledge of the number of species $n$ and their trait values, bypassing the need to infer a detailed phylogenetic tree.  For this, we will need a model for phenotypic trait evolution on the branches of an unknown phylogenetic tree generated by a Yule process.  

The simplest and most popular model for evolution of continuous phenotypic traits on phylogenetic trees is Brownian motion \citep{Felsenstein1985Phylogenies}.  Under this model, trait increments over a branch of length $t$ are normally distributed with mean $0$ and variance $\sigma^2 t$.  The trait values for extant species at the present time are observed as the vector $\mathbf{X}=(X_1,\ldots,X_n)$.  For a given topology $\tau$ with $n$ taxa and branching times $\mathbf{t}=(t_2,\ldots,t_{n-1})$, with tip data $\mathbf{X}$ generated on the branches of this tree by zero-mean Brownian motion with variance $\sigma^2$, the tip data are distributed according to a multivariate normal random variable.  More formally, 
\begin{equation}
\mathbf{X} \sim N\big(\mathbf{0}, \sigma^2 \mathbf{C}(\tau,\mathbf{t})\big) ,
\label{eq:traitnormal}
\end{equation}
where the entries of the variance-covariance matrix $\mathbf{C}(\tau,\mathbf{t})=\{c_{ij}\}$ are defined as follows: $c_{ii}=t$, and $c_{ij}$ is the time of shared ancestry for taxa $i$ and $j$, where $i\neq j$.  \citet{OMeara2006Testing} introduces \emph{disparity}, the sample variance of the tip data $\mathbf{X}$, 
\begin{equation}
 \text{disparity}(\mathbf{X}) = \frac{1}{n} (\mathbf{X} - \bar{X})' (\mathbf{X} - \bar{X}) ,
\end{equation}
where $\bar{X}$ is the mean of the elements of $\mathbf{X}$.  The expectation of the disparity, conditional on the tree topology $\tau$, branching times $\mathbf{t}$, and the Brownian variance $\sigma^2$, is
\begin{equation}
\begin{split}
 \E_\mathbf{X}(\text{disparity} \mid \tau,\mathbf{t},\sigma^2) &= \frac{1}{n} \E_\mathbf{X}\big((\mathbf{X} - \bar{X})' (\mathbf{X} - \bar{X}) \mid \tau,\mathbf{t}, \sigma^2\big) \\
  &= \sigma^2 \left[\frac{\text{tr}\big(\mathbf{C}(\tau,\mathbf{t})\big)}{n} - \frac{1}{n^2} \mathbf{1}' \mathbf{C}(\tau,\mathbf{t}) \mathbf{1} \right] \\
  &= \sigma^2 \left[t - \frac{1}{n^2} \mathbf{1}' \mathbf{C}(\tau,\mathbf{t}) \mathbf{1} \right] \\
  &= \sigma^2 \left[t - \frac{1}{n^2} \left(nt + 2\sum_{i=1}^n\sum_{j<i} c_{ij} \right)\right] \\
  &= \sigma^2 \left[\left(1-\frac{1}{n}\right)t - \frac{2}{n^2}\sum_{i=1}^n\sum_{j<i} c_{ij} \right]  ,
\end{split}
\label{eq:omeara}
\end{equation}
where we use the notation $\E_\mathbf{X}$ to indicate that the expectation is taken over realizations of the Brownian process that generates $\mathbf{X}$.  The fourth line above arises since the matrix $\mathbf{C}(\tau,\mathbf{t})$ is symmetric and every element on the diagonal is $t$.  However, every entry $c_{ij}$ is either zero or a branching time from the vector $\mathbf{t}$, so the nonzero terms in the sum consist of branching times $t_k$.  Let $z_k$ be the coefficient multiplying the branching time $t_k$ in the last line of \eqref{eq:omeara}.  Then we can express expected disparity as a weighted sum of the branching times,
\begin{equation}
 \E_\mathbf{X}(\text{disparity} \mid \tau,\mathbf{t},\sigma^2) = \sigma^2 \sum_{k=2}^n z_k t_k .
\label{eq:disparity}
\end{equation}
Figure \ref{fig:cov} illustrates how tree topology determines the matrix $\mathbf{C}(\tau,\mathbf{t})$ and expected disparity.

\begin{figure}
\begin{center}
\includegraphics{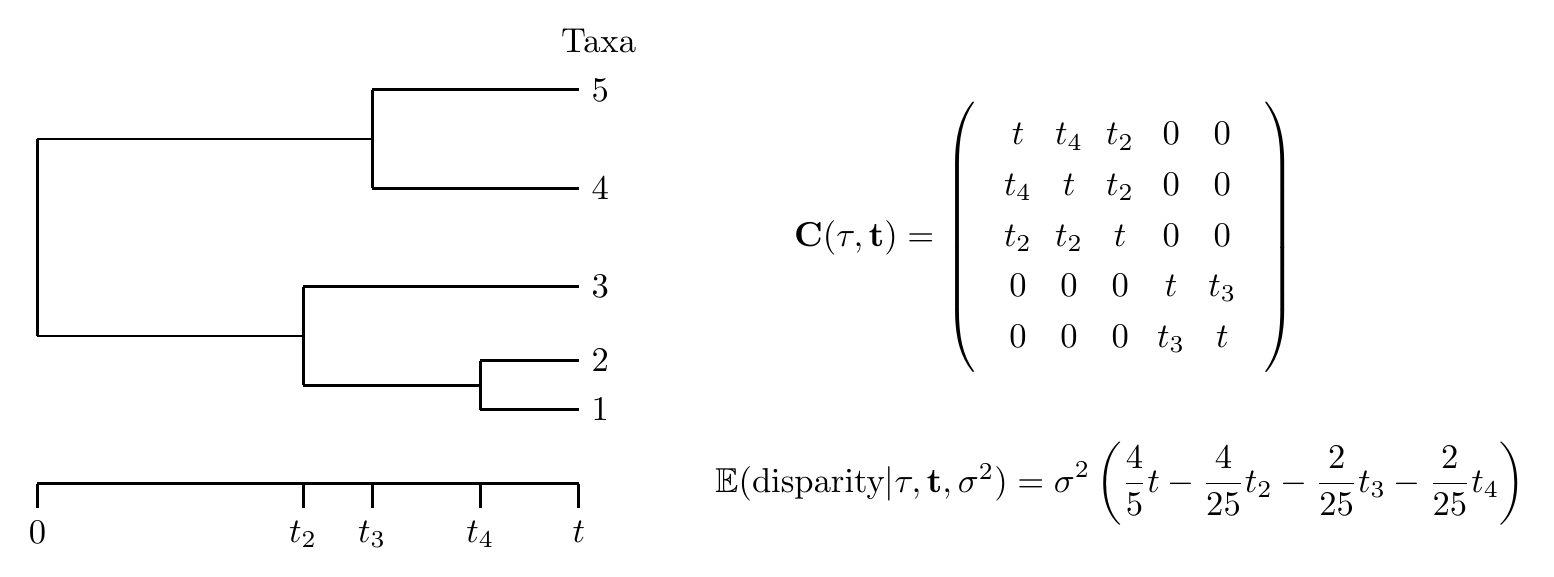}
\end{center}
\caption[How tree topology determines the phylogenetic covariance matrix]{How tree topology determines the matrix of Brownian covariances and expected disparity.  At left, a tree topology $\tau$ of crown age $t$ has $5$ taxa and branch times $\mathbf{t}=(t_2,t_3,t_4)$, where $t_k$ is the time of the branch from $k$ to $k+1$ lineages.  At right, the corresponding matrix $\mathbf{C}(\tau,\mathbf{t})$ of Brownian covariances.  The diagonal entries of $\mathbf{C}(\tau,\mathbf{t})$ are all $t$.  The $(i,j)$th entry of $\mathbf{C}(\tau,\mathbf{t})$ is the time of shared ancestry between taxa $i$ and $j$, for $i\neq j$.  For example, taxa $1$ and $2$ share ancestry for time $t_4$.  Expected disparity (using Brownian variance $\sigma^2$) is calculated using \eqref{eq:omeara}. Trait disparity cannot accumulate when there is only one species, so we draw the tree $\tau$ beginning with two lineages at time 0. }
\label{fig:cov}
\end{figure}

\subsection{Expected disparity as an accumulated reward}

The expected disparity \eqref{eq:disparity} has features in common with PD, since it is a scalar quantity that accumulates over the branches of the tree from time 0 to $t$.  The difference is that disparity implicitly incorporates tree-topological factors, which enter \eqref{eq:disparity} as weights in the sum of the branch lengths.  In addition, a Yule tree accumulates PD even when there is a single lineage, but the same is not true for disparity.  We develop these ideas in greater detail in this section.

Our goal is to express \eqref{eq:disparity} as a Markov reward process in a form equivalent to \eqref{eq:Rtdiff},
\begin{equation}
 R_t = \sigma^2 \sum_{k=1}^n a_k (t_k - t_{k-1}).
\label{eq:disparityreward}
\end{equation}
From \eqref{eq:omeara} we see that $a_n = z_n = \left(1-\frac{1}{n}\right)$, and $a_{n-1}$ can be found using $a_n$ and $z_{n-1}$, and so on.  We can formalize this recursive solution for the rewards by equating \eqref{eq:disparity} and \eqref{eq:disparityreward} as follows:
\begin{equation}
\sum_{k=2}^n z_k t_k =\sum_{k=1}^n a_k (t_k - t_{k-1}) 
 = a_nt_n + \sum_{k=1}^{n-1} (a_k - a_{k+1})t_k .
\end{equation}
Then recursively solving for the $a_k$'s gives $a_1=0$ and 
\begin{equation}
 a_k = \sum_{j=k}^n z_j.
\end{equation}
for $k=2,\ldots,n$.  Now defining $R_t(\mathbf{a})$ to be the Yule reward process with rewards $\mathbf{a}=(a_1,\ldots,a_n)$ under the topology $\tau$, the expected disparity is distributed as
\begin{equation}
\E_\mathbf{X}(\text{disparity}\mid \tau,\lambda,\sigma^2,n) \sim \sigma^2 R_t(\mathbf{a}) .
\label{eq:aMz} 
\end{equation}
Note that we no longer need to condition on the branch lengths $\mathbf{t}=(t_1,\ldots,t_n)$ in the expected disparity -- they have been ``integrated out''.  Therefore, to find the distribution of expected trait variance under a Brownian motion process on a Yule tree with topology $\tau$, we need only find the relevant rewards $\mathbf{a}$ and compute the corresponding distribution of $R_t(\mathbf{a})$.  

As a concrete example, consider the five-taxon tree in Figure \ref{fig:cov}.  The expected disparity, given this topology $\tau$ and arbitrary branch lengths $\mathbf{t}=(t_2,t_3,t_4)$, is
\begin{equation}
\E(\text{disparity}|\tau,\mathbf{t},\sigma^2) = \sigma^2\left[\frac{4}{5}t - \frac{4}{25}t_2 - \frac{2}{25}t_3 - \frac{2}{25}t_4\right] .
\end{equation}
The coefficients are given by
\begin{equation}
\mathbf{z} = \left(-\frac{4}{25},\ -\frac{2}{25},\ -\frac{2}{25},\ \frac{4}{5}\right).
\end{equation}
Solving for the rewards $\mathbf{a}$, we obtain
\begin{equation}
 \mathbf{a} = \left(0, \frac{12}{25},\ \frac{16}{25},\ \frac{18}{25},\ \frac{4}{5} \right)
\end{equation}
which is easily verified by hand.  This leads us to our second Theorem, which gives an expression for the distribution of $R_t(\mathbf{a})$.
\begin{thm}
\label{thm:generalyule}
In a Yule process with rate $\lambda$ and arbitrary rewards $\mathbf{a} = (a_1,\ldots,a_n)$, the Laplace transform of $v_{mn}(x,t|\mathbf{a})$ is given by
\begin{equation}
 f_{mn}(r,t)  = \begin{cases}
    e^{-(m\lambda + a_mr)t} & m=n\text{, and} \\[1em]
 \displaystyle\lambda^{n-m} \frac{(n-1)!}{(m-1)!} \sum_{j=m}^n \frac{e^{-(j\lambda + a_jr)t}}{\prod_{k\neq j} \big(\lambda(k - j) + r(a_k - a_j)\big)} & n>m .
\end{cases}
\label{eq:generalyule}
\end{equation}
\end{thm}
\noindent The proof of this Theorem is given in Appendix \ref{app:generalyule}. To obtain the probability distribution of the accumulated reward, we must invert \eqref{eq:generalyule},
\begin{equation}
v_{mn}(x,t) = \mathscr{L}^{-1}\big[ f_{mn}(r,t) \big](x).
\end{equation}
For $m=n$ and $n=m+1$, there are simple expressions for the inverse Laplace transform.  Under certain conditions on the rewards $\mathbf{a}$, there is a straightforward analytic inversion of \eqref{eq:generalyule} for general $n>m$, but the rewards computed using the times of shared ancestry in a phylogenetic tree do not always satisfy these conditions.  Therefore, it is often easier to numerically invert \eqref{eq:generalyule}; we discuss this issue in much greater detail in Appendix \ref{app:inversion}, and provide a straightforward method for numerical inversion of the Laplace transform \eqref{eq:generalyule} based on the method popularized by \citet{Abate1995Numerical}.


\subsection{Approximate likelihood and inference for $\sigma^2$}

We now describe a statistical procedure for using Theorem \ref{thm:generalyule} to perform statistical for the unknown Brownian variance $\sigma^2$.  Suppose that in a clade of $n$ species we have a crude tree topology $\tau$ (without branch lengths).  This topology could be derived from parsimony, distance-based tree reconstruction methods, or one could simply use family/genus/species information to assign a hierarchy of relationships and resolve polytomies randomly.  Given the tree topology $\tau$, one can compute the rewards $\mathbf{a}$.  Suppose also that we have calculated trait disparity for each of $J$ independent continuous quantitative traits that arise from Brownian motion on the branches of the unknown phylogenetic tree, starting at the root.  Let 
\begin{equation}
D_n^{(j)} = \frac{1}{n}\left(\mathbf{X}^{(j)} - \bar{X}^{(j)}\right)' \left(\mathbf{X}^{(j)} - \bar{X}^{(j)}\right) ,
\end{equation}
be the observed disparity for the $j$th phenotypic trait, where $\mathbf{X}^{(j)}$ is the vector of $n$ trait values for the $j$th phenotypic trait and $\bar{X}^{(j)}$ is the mean of the elements of $\mathbf{X}^{(j)}$.  Then we calculate the mean disparity $\bar{D}_n$ across these $J$ traits:
\begin{equation}
 \bar{D}_n = \frac{1}{J}\sum_{j=1}^J D_n^{(j)} .
\end{equation}
Note that in order to find $\bar{D}_n$, we do not need the individual trait measurements themselves -- only the disparities.  Then by the law of large numbers, $\bar{D}_n \to \E(D_n)$ as $J\to\infty$, where $D_n$ is the asymptotic mean disparity across all possible traits.  Therefore, we approximate the distribution of $\bar{D}_n$ as follows:
\begin{equation}
\begin{split}
\bar{D}_n &\approx \E(D_n) \\
          &\sim \E_\mathbf{X}(\text{disparity}\mid \tau,\sigma^2,\lambda,n,t) \\
          &= \sigma^2 R_t(\mathbf{a}) .
\end{split}
\label{eq:Dnbar}
\end{equation}
where $\mathbf{a}$ is the vector of rewards obtained from the topology $\tau$.  This approximate relation provides the connection between observable mean trait disparity and the probability distribution in Theorem \ref{thm:generalyule} that we need in order to compute the probability of the observed disparities.  Suppose the stem age of an unresolved tree is $t$, and let
\begin{equation}
 f_Y(x) = \frac{v_{mn}(x|t,\lambda,\mathbf{a})}{P_{mn}(t)} 
\end{equation}
be the distribution of expected disparity in a Yule process with general rewards $\mathbf{a}$, conditional on $Y(0)=m$ and $Y(t)=n$.   Here, $x$ is the expected trait disparity, which we approximate by our observed (and therefore fixed) $\bar{D}_n$.  From \eqref{eq:Dnbar}, we write 
\begin{equation}
 \frac{\bar{D}_n}{\sigma^2} \sim R_t(\mathbf{a})
\label{eq:approxrv}
\end{equation}
so the likelihood is approximately 
\begin{equation}
f_Y(\bar{D}_n/\sigma^2) .
\label{eq:approxlik}
\end{equation}
Finally, we propose the approximate maximum likelihood estimator 
\begin{equation}
 \hat{\sigma}^2 = \argmax{\sigma^2} f_Y(\bar{D}_n/\sigma^2) .
\label{eq:sigma2hat}
\end{equation}
To find $\hat{\sigma}^2$, note that in a Yule reward process in which $a_k<a_j$ for $k<j$, the value of the reward is constrained to lie in the interval
\begin{equation}
 a_m t\leq R_t(\mathbf{a}) \leq a_n t
\label{eq:rewardinterval}
\end{equation}
where we have assumed $Y(0)=m$ and $Y(t)=n$.  Additionally, when none of the rewards $a_k$ are zero we are justified in dividing $\bar{D}_n$ by the terms in \eqref{eq:rewardinterval} to obtain bounds on the possible value of $\hat{\sigma}^2$ that maximizes \eqref{eq:approxlik},
\begin{equation}
 \frac{\bar{D}_n}{a_nt} \leq \hat{\sigma}^2  \leq \frac{\bar{D}_n}{a_mt} .
\label{eq:sigma2hatinterval}
\end{equation}
When $m=1$ and $a_1=0$, the upper bound is infinity.  However, in practice when $m=1$ and $n\geq 4$, it is safe to assume that 
\begin{equation}
 \frac{\bar{D}_n}{a_nt} \leq \hat{\sigma}^2  \leq \frac{\bar{D}_n}{a_2t} .
\label{eq:sigma2hatinterval2}
\end{equation}
We solve \eqref{eq:sigma2hat} using the numerical Newton-Raphson method provided by the \texttt{R} function \texttt{nlm}.

We emphasize that \eqref{eq:sigma2hat} is not a traditional maximum likelihood estimator.  We have approximated the distribution of $\bar{D}_n$ by the distribution of expected disparity, giving an approximate likelihood \eqref{eq:approxlik} that may not attain its maximum at the same value of $\sigma^2$ as the true likelihood.  In addition, the density \eqref{eq:approxlik} is non-differentiable at several points, and for $n=m+1$ attains its maximum at the lower boundary $R_t(\mathbf{a})=a_mt$ (see Figure \ref{fig:yulerewards} with $m=1$, $n=2$).  These issues complicate application of traditional asymptotic theory for maximum likelihood estimates, and the classical large sample theory may not apply because the likelihood is only an approximation.  One consequence of the violation of these traditional assumptions is that we are unable to provide meaningful standard errors for $\hat{\sigma}^2$ using only the approximate likelihood \eqref{eq:approxlik}.


\subsection{Simulations}

To empirically evaluate the correctness of the analytic distributions we derived in Theorem \ref{thm:generalyule}, we simulated trait data via Brownian motion on trees generated by a Yule process with age $t=1$ and branching rate $\lambda=1$.  For $n=3,\ldots,8$, we chose one tree topology and simulated $N_\text{btimes}=2000$ sets of branch lengths from a Yule process for $n$ species \citep{Stadler2011Simulating}.  For each set of branching times, we simulated $N_\text{BM}=2000$ realizations of Brownian motion with $\sigma^2=1$ to generate trait values at the tips of the tree \citep{Paradis2004APE}.  For each of the 2000 sets of tip values, we calculated the mean disparity using \eqref{eq:omeara}.  Figure \ref{fig:sim} shows histograms of the mean disparities with the analytic distribution $f_Y(x)$ overlaid, with good correspondence.  The tree topology (with arbitrary branch lengths for display) is shown in gray above each histogram. 

\begin{figure}
\centering
\includegraphics[width=\textwidth]{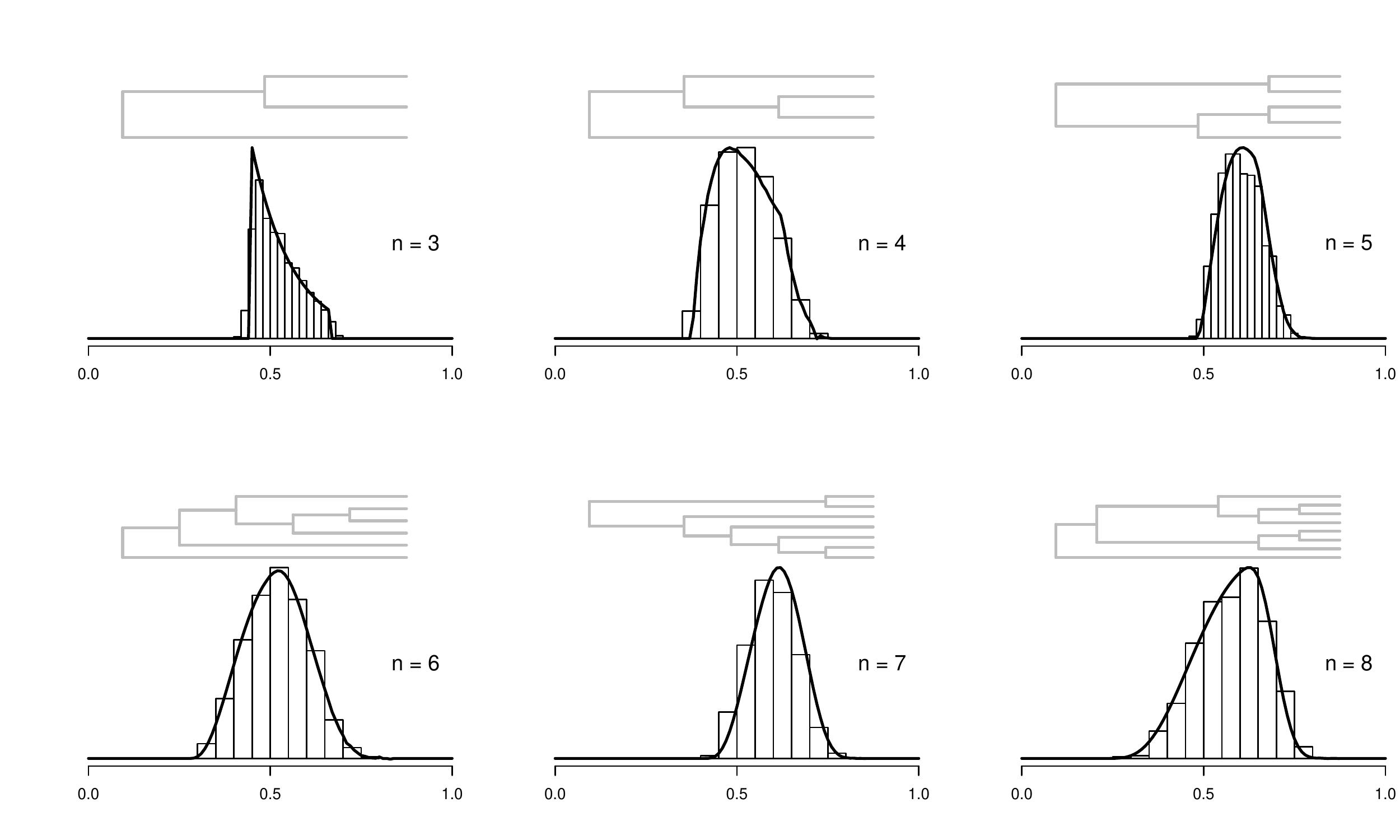}
\caption[Correspondence between simulated and calculated PD distributions]{Empirical correspondence between the derived expressions for the distribution of expected disparity and simulated mean disparity histograms for trees with different numbers of taxa.  For each $n=3,\ldots,8$, we simulated a single tree topology (shown above each histogram in gray).  We then simulated 2000 sets of branch lengths for this topology under the Yule process.  For each set of branch lengths, we calculated the mean disparity from 2000 simulations of zero-mean Brownian motion with variance $\sigma^2=1$ on this tree.}
\label{fig:sim}
\end{figure}

To evaluate our estimation methodology, we take a similar approach, but for each simulated set of mean disparities, we infer $\hat{\sigma}^2$, an approximate maximum likelihood estimate of $\sigma^2$.  Figure \ref{fig:simestimates} shows estimates of $\sigma^2$ for different species richness $n$ under different simulation conditions.  For $n=3,\ldots,12$, we generated 100 trees, each with $N_\text{btimes}=1,5,10,100$ sets of branching times.  For each set of branching times, we evolved $N_\text{BM}=1,5,10,100$ traits by Brownian motion along the branches with rate $\sigma^2=1$ and computed the mean disparity.  Then, given the $N_\text{btimes}$ mean disparities, we maximized the approximate likelihood to find $\hat{\sigma}^2$.  Each dot in Figure \ref{fig:simestimates} represents one estimate, and the dots are jittered slightly to show their density.  The variance in the estimator is large when the number of simulated branching time sets and Brownian realizations is small since \eqref{eq:Dnbar} and hence \eqref{eq:approxrv} become poor approximations to the mean disparity.  However, the approximate maximum likelihood estimator for $\sigma^2$ appears to have the desirable property of statistical consistency: the deviation of the estimates from the true value $\sigma^2=1$ goes to zero as the number of mean disparity observations becomes large.


\begin{figure}
\centering

\includegraphics{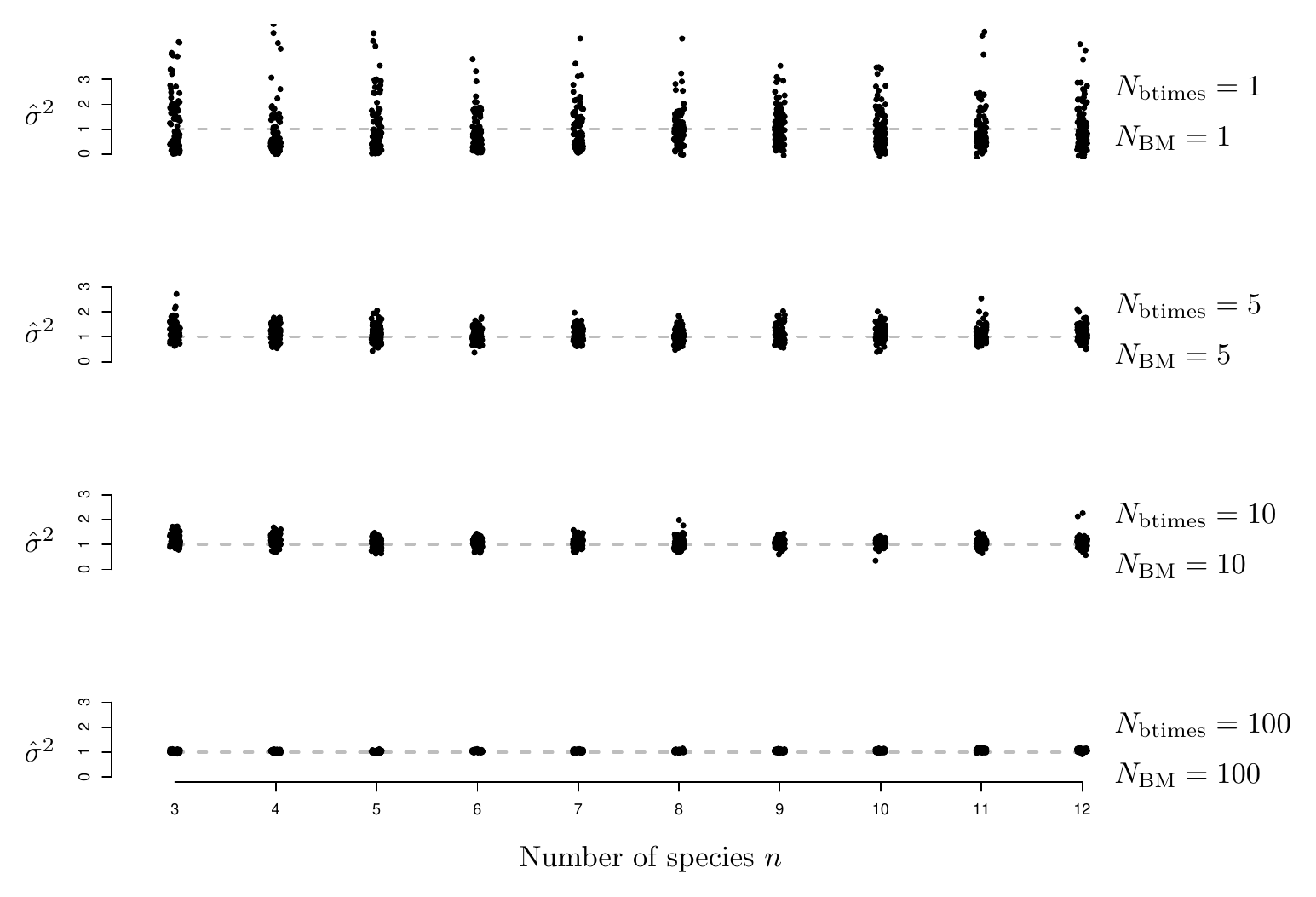}

\caption[Empirical consistency of maximum likelihood estimates]{Empirical consistency of approximate maximum likelihood estimates of $\sigma^2$. The number of species $n$ in the unobserved phylogenetic tree is shown on the horizontal axis.  Each set of plots shows 100 estimates of $\sigma^2$ from $N_\text{btimes}$ simulations of different branching times under a Yule process with $n$ species and from $N_\text{BM}$ independent realizations of Brownian motion used to compute the mean disparity for each set of branching times.  The estimates are jittered to show the sampling distribution.  The gray dotted line shows the true value $\sigma^2=1$.}
\label{fig:simestimates}
\end{figure}


\section{Application to evolution of body size in the order Carnivora}

To illustrate the usefulness of our method in practical comparative inference, we estimate thirteen family-wise Brownian variance rates for body size evolution in the mammalian order Carnivora using observed log-body size disparities and species richness information \citep{Gittleman1986Carnivore,Gittleman1998Body,Slater2012Fitting}. Carnivora includes members with very large and small body masses, including wide diversity within individual families \citep{Nowak1999Walkers}.  We included only families with 2 or more species, since families with only one species do not reveal useful information about intra-family Brownian variance. The dataset, comprising 284 species, included the families Canidae, Eupleridae, Felidae, Herpestidae, Hyaenidae, Mephitidae, Mustelidae, Otariidae, Phocidae, Prionodontidae, Procyonidae, Ursidae, and Viverridae.  Figure \ref{fig:carnivoratree} shows the backbone phylogeny (from \citet{Eizirik2010Pattern}) and the unresolved clades.  Our analysis takes advantage of utilities for manipulating trees and quantitative trait data in the \texttt{ape} package (described in \citet{Paradis2004APE}) and simulating trees with the \texttt{TreeSim} package (described in \citet{Stadler2011Simulating}), using the statistical programming language \texttt{R} \citep{CRAN2012}.  We intentionally limit our analysis to the Brownian variance in each family and use only body size disparity in order to demonstrate the simplicity of our approximate method under conditions of very little data.

\begin{figure}
\centering
\includegraphics[width=1.0\textwidth]{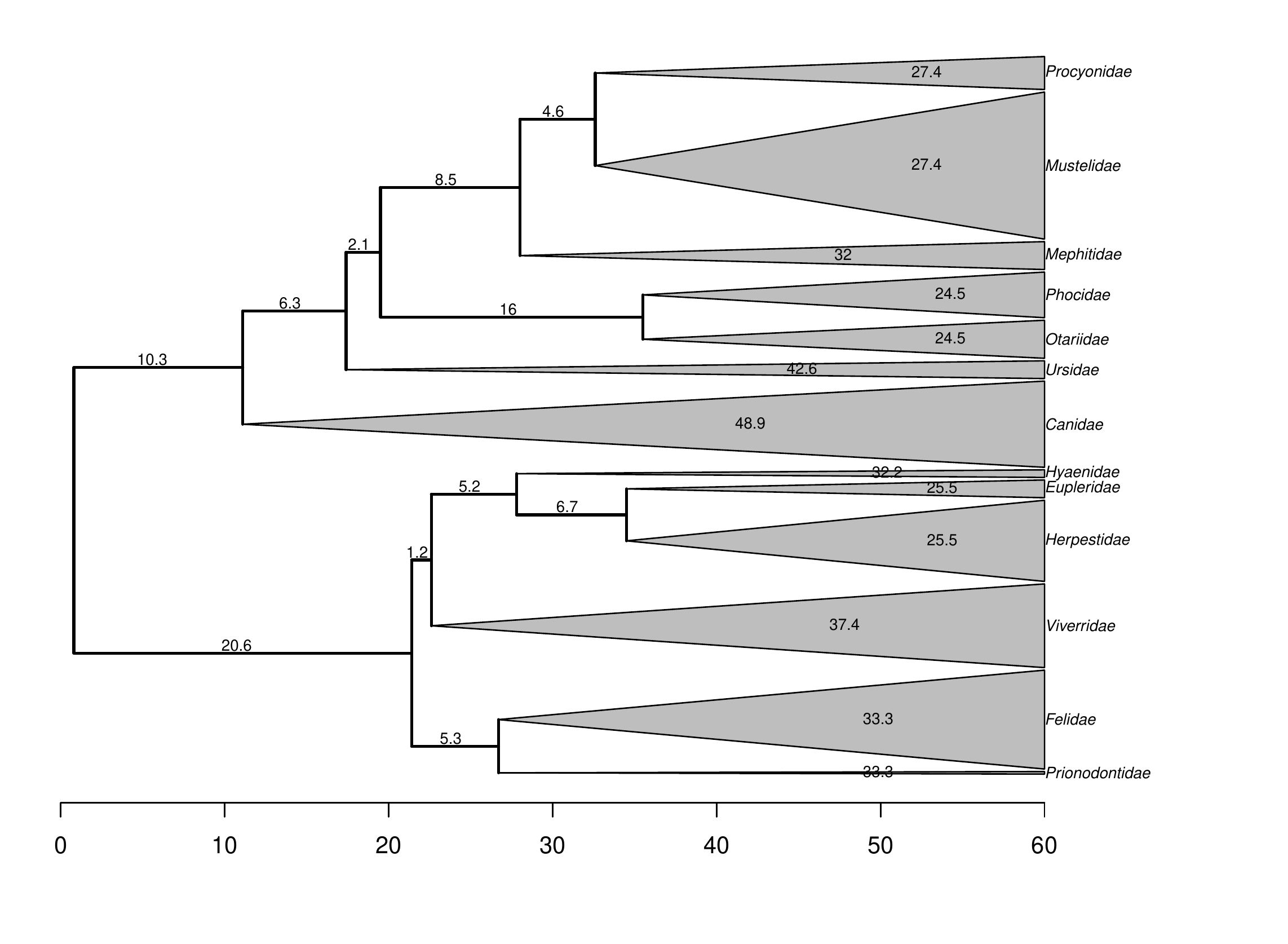}
\caption[A family-level phylogenetic tree for order Carnivora]{A family-level phylogenetic tree for order Carnivora.  The phylogeny within each family, shown here as a gray triangle, is not known with certainty.  The length of the base of the gray triangles represents the number of species in the family.  The ``backbone'' tree connecting the unresolved clades is assumed fixed.  Branch lengths are shown along each branch. }
\label{fig:carnivoratree}
\end{figure}

The first step is to estimate the speciation rate $\lambda$ from the backbone tree and the unresolved clades.  Even though the tree is unobserved within each family, we can still find the exact maximum likelihood estimate of $\lambda$ for the tree as a whole.  On a fully resolved branch of length $t$, the log-likelihood of $\lambda$ is $\log(\lambda) - \lambda t$.  For an unresolved clade of age $t$ with one species at the crown which grows to include $n$ species, the log-likelihood from \eqref{eq:yuleprob} is proportional to $-\lambda t + (n-1)\log(1-e^{-\lambda t})$.  Summing these partial log-likelihoods over the whole tree gives the log-likelihood for $\lambda$; maximizing this function, we find that $\hat{\lambda}=0.069$ per million years.  In what follows, we assume that $\lambda=\hat{\lambda}$.  To each unresolved family clade we associate the clade disparity for body size.  Our analysis will consist of estimating the family-wise Brownian variance $\sigma_j^2$ for the $j$th family, assuming one species at the stem age for each clade shown in Figure \ref{fig:carnivoratree}.  In this way, we integrate over crown ages for each family under the Yule model.  To apply the methodology for modeling expected trait disparity developed in section \ref{sec:disparity}, we regard the unresolved clades as ``soft polytomies'' in which branch lengths are unknown \citep{Purvis1993Polytomies}.  We therefore resolve these polytomies randomly without assigning branch lengths.  

Our analysis of the Carnivora family-wise Brownian variances takes approximately 30 seconds to run on a laptop computer.  However, to evaluate the variability in our estimates, we ran the analysis 100 times with randomly resolved polytomies for each family.  Table \ref{table:carnivora} shows the family name, species richness ($n$), stem age ($t$), observed body size disparity ($\bar{D}_n$), the mean estimate of $\sigma_j^2$, and the approximate standard error of the estimate for each family.  We calculated standard errors as the empirical standard deviation of the 100 Brownian variance estimates.  Our estimates of $\sigma^2$ reveal readily interpretable information about the evolution of body size in each family that is not available directly from observed disparities alone.  For example, the families Herpestidae and Viverridae have almost identical species richness, but quite different disparity measurements. Perhaps surprisingly, we have estimated nearly equal Brownian variances for the two families.  Why does our method produce such similar estimates of $\sigma^2$?  The answer lies in the ages of the clades -- Viverridae is almost 50\% older than Herpestidae.  Two clades with the same richness whose traits evolve by Brownian motion at the same rate can exhibit very different disparity measurements, depending on their ages.  This example illustrates how our method provides an approximate way to untangle the complex interaction of time, species, and observed trait variance (in the words of \citet{Ricklefs2006Time}) for unresolved clades.

\begin{table}
\begin{center}
\begin{tabular}{lccccc}
  \hline
Family & Richness & TMRCA & Disparity & $\hat{\sigma}^2$ & SE \\ 
  \hline
  Prionodontidae & 2 & 33.30 & 0.131 & 0.051 & 0 \\ 
  Felidae & 40 & 33.30 & 1.588 & 0.118 & 0.094 \\ 
  Viverridae & 34 & 37.40 & 0.606 & 0.034 & 0.018 \\ 
  Herpestidae & 33 & 25.50 & 0.482 & 0.035 & 0.017 \\ 
  Eupleridae & 8 & 25.50 & 0.916 & 0.081 & 0.010 \\ 
  Hyaenidae & 4 & 32.20 & 0.805 & 0.084 & 0.001 \\ 
  Canidae & 35 & 48.90 & 0.678 & 0.031 & 0.012 \\ 
  Ursidae & 8 & 42.60 & 0.303 & 0.025 & 0.002 \\ 
  Otariidae & 16 & 24.50 & 0.386 & 0.029 & 0.005 \\ 
  Phocidae & 19 & 24.50 & 0.751 & 0.065 & 0.030 \\ 
  Mephitidae & 12 & 32.00 & 0.570 & 0.038 & 0.004 \\ 
  Mustelidae & 59 & 27.40 & 2.263 & 0.220 & 0.239 \\ 
  Procyonidae & 14 & 27.40 & 0.531 & 0.038 & 0.006 \\ 
   \hline
\end{tabular}
\end{center}
\caption[Estimates of Brownian variance for Carnivora body size data]{Species richness, TMRCA, body size disparity, and estimated Brownian variance $\hat{\sigma}^2$ for each family in the Carnivora dataset.  Note that Canidae and Herpestidae have very different disparity measurements, but nearly identical estimates of $\sigma^2$.  This discrepancy is due to the difference in their ages; we explain the interaction between time, species number and disparity in greater detail in the text.}
\label{table:carnivora}
\end{table}



\section{Discussion: Comparative phylogenetics without trees?}


In this paper we have outlined a method for integrating over Yule trees.  We presented an expression for the distribution of PD in an unresolved tree, conditional on the number of species $n$ and age $t$.  We showed that the expected disparity can be represented in a similar way as PD, since it accumulates along the branches of a phylogenetic tree.  We also derived a statistical framework that uses a very small amount of information ($n$, $t$, and $\bar{D}_n$) for an unresolved clade to derive a meaningful estimator for the Brownian rate $\sigma^2$.  It may seem counterintuitive that one can estimate the Brownian rate for an unresolved tree with $n$ taxa, given a single disparity measurement.  However, the structure provided by the Yule process allows this inference by providing just enough information about the distribution of branching times that generate the tree to model the average phenotypic disparity under Brownian motion.  This permits analytic integration over two random objects: the collection of branching times of the tree and realizations of Brownian motion.  Three assumptions make this possible: first, we fix a topology $\tau$ without branch lengths; second, we assume that branch lengths come from a Yule process; third, we compute the distribution of expected disparity, which is a scalar quantity that encapsulates the most important information in the covariance matrix $\mathbf{C}(\tau)$.  In exchange for these assumptions, we gain what frustrated reviewers of Felsenstein's paper apparently wished for: an estimator that ``obviate[s] the need to have an accurate knowledge of phylogeny'' \citep{Felsenstein1985Phylogenies}.  Whether these assumptions are warranted depends fundamentally on the scientific questions at hand, and the available data.  

Perhaps the most satisfying use of our method is in providing an approximate and model-based answer to the questions posed by \citet{Ricklefs2006Time} and \citet{Bokma2010Time}, in their similarly-named papers.  Our answer is approximate because it substitutes an observation for an expectation in \eqref{eq:Dnbar}; it is model-based because we assume that trees arise from a Yule process and traits evolve Brownian motion.  Equation \eqref{eq:Dnbar} expresses a heuristic relationship explaining the origins of phenotypic disparity, which we reproduce here for emphasis:
\begin{equation}
 \bar{D}_n \sim \sigma^2 R_t(\mathbf{a}) .
\label{eq:Dnbar2}
\end{equation}
On the left-hand side is the observed disparity.  On the right-hand side, $R_t(\mathbf{a})$ serves as a scalar summary of tree shape --  it depends only on $n$, clade age $t$, and tree topology $\tau$.  This reveals that even when we restrict our attention to expected disparity under the simplest evolutionary models, the interaction between $n$, $t$ and the branching structure of the tree in $R_t(\mathbf{a})$ is complex, but the Brownian variance simply scales the tree-topological term.  We see that $\bar{D}_n$ scales linearly with $\sigma^2$ when $n$, $t$, and the topology $\tau$ are fixed.  However, changing one of $n$, $t$, or $\tau$ while holding $\sigma^2$ constant will induce a nonlinear change in $\bar{D}_n$.  We conclude that it is not possible to partition the time-dependent and speciation-dependent influences on the accumulation of trait variance in a simple way as suggested by \citet{Ricklefs2006Time} under the stochastic models we study in this paper.

As an inducement to spur research on analytic integration over trees, \citet{Bokma2010Time} offers a monetary reward for an expression for the distribution of sample variance from a birth-death tree with Brownian trait evolution on its branches.  We have solved a simpler version of Bokma's challenge by providing the distribution of expected trait variance for a specific topology under a pure-birth process.  The expression \citet{Bokma2010Time} seeks is difficult to find for two reasons: first, it would require analytic integration over discrete tree topologies; second, and more intuitively, integrating over both topologies and Brownian realizations would subsume the Brownian variance $\sigma^2$ on the right-hand side of \eqref{eq:Dnbar} into a nonlinear term that depended on $n$, $t$, and $\sigma^2$ in a very complicated way.  Alternatively, simulation-based approaches provide an appealing alternative method to integrate over trees and Brownian motions without requiring approximation of the disparity by its expectation.  Indeed, Bayesian methods exist to sample from the distribution of Yule trees, conditional on observed trait values at the tips, thereby providing both estimates of Brownian rates and phylogenies simultaneously, while using all the available trait data \citep{Drummond2012Bayesian}.

Tree-free comparative evolutionary biology comes at a price -- there are several important drawbacks to our approach.  First, even under the Yule model for speciation (with the correctly specified branching rate $\lambda$) and zero-mean Brownian motion for traits, integrating over all possible Yule trees introduces great uncertainty in estimates of $\sigma^2$.  Figure \ref{fig:simestimates} illustrates this issue: while the estimates of $\sigma^2$ eventually converge to the true value as the number of branch length and trait realizations becomes large, the variance in these estimates can be substantial for smaller datasets.  Furthermore, the assumption that $\bar{D}_n \sim \sigma^2 R_t(\mathbf{a})$ may be suspect if the number of traits analyzed is small enough that the mean trait disparity is a poor substitute for the expected disparity.  

We conclude with a mixed message about analytic integration over trees.  First, it is possible to derive meaningful estimators for parameters of interest under simple evolutionary models, if one is willing to make assumptions about the mean behavior of the models.  The estimates are usually reasonable, and may provide valuable insight into the basic properties of evolutionary change under these models -- even our simplistic analysis of Carnivora body size evolution reveals the complex interaction of clade age, species number, and evolutionary rate.   These estimates may be useful as starting points for more time-consuming simulation analyses.  Second, and more pessimistically, sophisticated analytic methods for integrating over trees cannot conjure evolutionary information from the data that is not there already.    As evolutionary biologists further refine our knowledge of the tree of life, the number of clades whose phylogeny is truly unknown may diminish, along with interest in tree-free estimation methods.



\section*{Acknowledgements}

We are extremely grateful to Michael Alfaro and Graham Slater for introducing us to the problem of finding the distribution of quantitative trait variance, providing the Carnivora dataset, and for helpful comments on the manuscript.  John  Welch and Eric Stone provided insightful criticism and suggestions.  FWC was supported by NIH grant T32GM008185, and MAS was supported by NIH grants R01 GM086887, HG006139,  and NSF grant DMS-0856099.


\appendix

\section{Markov rewards for Yule processes}

\label{app:reward}

In this Appendix, we prove one lemma and the two Theorems presented in the text.  In the first proof, we derive a representation of the forward equation for a Yule reward process.  Our development follows that given by \citet{Neuts1995Algorithmic}. 

\newtheorem{lem}{Lemma}
\begin{lem}
\label{lem:rewardodes}
In a Yule reward process $R_t = \int_0^t a_{Y(s)}\dx{s}$ with arbitrary positive rewards $a_1,a_2,\ldots$, the Laplace-transformed reward probabilities satisfy the ordinary differential equations
\begin{equation}
 \od{f_{mn}(r,t)}{t} = -(n\lambda+a_n r)f_{mn}(r,t) + (n-1)\lambda f_{m,n-1}(r,t)  .
\end{equation}
\end{lem}

\begin{proof}
Let $V_{mn}(x,t) = \Pr(R_t \leq x, Y(t)=n\mid Y(0)=m)$.  We can re-write this quantity in a more useful form by conditioning on the time of departure $u$ from state $m$, noting that the accumulated reward is $a_m u$, and then integrating over $u$.  If $m=n$ and no departure occurs, the accumulated reward is $a_m t$.  Putting these ideas together, we obtain
\begin{equation}
\begin{split}
 V_{mn}(x,t) &= \Pr(R_t\leq x, Y(s)=m\text{ for }0\leq s\leq t) \\
             & \quad + \int_0^t \Pr(m\to m+1\text{ at time } u) V_{m+1,n}(x-a_m u,t-u) \dx{u} \\
  &= \delta_{mn}e^{-m\lambda t} H(x-a_mt) + \int_0^t m\lambda e^{-m\lambda u}V_{m+1,n} (x-a_m u,t-u) \dx{u} .
\end{split}
\end{equation}
Now consider the Laplace transform $F_{mn}(r,t)$ of $V_{mn}(x,t)$, with respect to the reward variable $x$,
\begin{equation}
\begin{split}
 F_{mn}(r,t) &= \mathscr{L}\big[V_{mn}(x,t)\big](r) \\
 &= \int_0^\infty e^{-rx} V_{mn}(x,t) \dx{x} \\
  &= \delta_{mn}\frac{e^{-(m\lambda+a_m r)t}}{r} + \int_0^t m\lambda e^{-m\lambda u}\int_{a_m u}^\infty e^{-rx} V_{m+1,n} (x-a_m u,t-u)\dx{x} \dx{u} .
\end{split}
\end{equation}
Making the substitution $y=x-a_m u$ in the Laplace integral, we have 
\begin{equation}
\begin{split}
 F_{mn}(s,t) &= \delta_{mn} \frac{e^{-(m\lambda+a_m r)t}}{r} + \int_0^t m\lambda e^{-m\lambda u}\int_0^\infty e^{-r(y+a_m u)} V_{m+1,n} (y,t-u)\dx{y} \dx{u} \\
  &= \delta_{mn}\frac{e^{-(m\lambda+a_m r)t}}{r} + \int_0^t m\lambda e^{-(m\lambda+a_m r)u} F_{m+1,n}(r,t-u) \dx{u} .
\end{split}
\end{equation}
Now multiplying both sides by $e^{(m\lambda+a_m r)t}$ and differentiating with respect to $t$, we obtain
\begin{equation}
\begin{split}
 \pd{}{t} \left[e^{(m\lambda+a_m r)t} F_{mn}(r,t)\right] &= \pd{}{t} \int_0^t m\lambda e^{(m\lambda+a_m r)(t-u)} F_{m+1,n}(r,t-u) \dx{u} \\
 &= \pd{}{t} \int_0^t m\lambda e^{(m\lambda+a_m r)u} F_{m+1,n}(r,u) \dx{u} .
\end{split}
\end{equation}
Expanding the left-hand side by the product rule and using the fundamental theorem of calculus on the right, we find that 
\begin{equation}
 e^{(m\lambda+a_m r)t} \left( (m\lambda+a_m r) F_{mn}(r,t) + \pd{}{t}F_{mn}(r,t)\right) = e^{(m\lambda+a_m r)t}  m\lambda F_{m+1,n}(r,t) .
\end{equation}
Cancelling common factors and rearranging, we obtain the Kolmogorov backward equation,
\begin{equation}
 \pd{}{t} F_{mn}(r,t) = - (m\lambda+a_m r)F_{mn}(r,t) + m\lambda F_{m+1,n}(r,t) .
\label{eq:recur}
\end{equation}
However,
\begin{equation}
 rF_{mn}(r,t) = \mathscr{L}\left[\pd{}{x} V_{mn}(x,t)\right](r) = \mathscr{L}\big[v_{mn}(x,t)\big](r) = f_{mn}(r,t) .
\end{equation}
Plugging $rF_{mn}(r,t) = f_{mn}(r,t)$ into \eqref{eq:recur}, we find that the $f_{mn}(r,t)$ satisfy the same system of ordinary differential equations,
\begin{equation}
 \pd{}{t} f_{mn}(r,t) = - (m\lambda+a_m r)f_{mn}(r,t) + m\lambda f_{m+1,n}(r,t) .
\label{eq:yulebackward}
\end{equation}
These are the \emph{backward equations} for the Laplace transformed reward process.  To solve \eqref{eq:yulebackward}, we note that any solution to the forward equations is a solution to the backward equations in a birth process \citep{Grimmett2001Probability}.  Therefore, \eqref{eq:yulebackward} is equivalent to the forward system
\begin{equation}
 \pd{f_{mn}(r,t)}{t} = - (n\lambda + a_n r) f_{mn}(r,t) + (n-1)\lambda f_{m,n-1}(r,t) 
\end{equation}
for $n=m,m+1,m+2,\ldots$  This completes the proof.
\end{proof}


\section{Proof of Theorem \ref{thm:yule}}

\label{app:yule}

\begin{proof} 
Lemma \ref{lem:rewardodes} with $a_k=k$ for $k=0,1,\ldots$ gives 
\begin{equation}
 \pd{f_{mn}(r,t)}{t} = - n(\lambda + r) f_{mn}(r,t) + (n-1)\lambda f_{m,n-1}(r,t) .
\label{eq:yulerewardlt}
\end{equation}
Define $g_{mn}(r,s)$ to be the Laplace transform of $f_{mn}(r,t)$ with respect to the time variable $t$.  Transforming \eqref{eq:yulerewardlt} gives
\begin{equation}
 sg_{mn}(r,s) - \delta_{mn} = - n(\lambda + r) g_{mn}(r,s) + (n-1)\lambda g_{m,n-1}(r,s) .
\end{equation}
Letting $m=n$, we find that
\begin{equation}
 g_{mm}(r,s) = \frac{1}{s+m(\lambda+r)} .
\end{equation}
Next, we form a recurrence and solve for $g_{mn}(r,s)$ to obtain
\begin{equation}
\begin{split}
 g_{mn}(r,s) &= \frac{(n-1)\lambda}{s+n(\lambda+r)} g_{m,n-1}(r,s) \\
 &= \frac{(n-1)\cdots m \lambda^{n-m}}{\prod_{j=m+1}^n \big[s+j(\lambda+r)\big]} g_{m,m+1}(r,s) \\
 &= \frac{(n-1)!}{(m-1)!} \frac{\lambda^{n-m} }{\prod_{j=m}^n \big[s+j(\lambda+r)\big]} .
\end{split}
\end{equation}
We proceed via a partial fractions decomposition of the product in the denominator above,
\begin{equation}
\begin{split}
 g_{mn}(r,s) &= \lambda^{n-m} \frac{(n-1)!}{(m-1)!} \sum_{j=m}^n \left(\prod_{k\neq j} (\lambda+r)(k-j) \right)^{-1} \frac{1}{s+j(\lambda+r)} \\
   &= \frac{(n-1)!}{(m-1)!} \lambda^{n-m} \sum_{j=m}^n \frac{\left[\left(\prod_{k=m}^{j-1}(k-j)\right)\left(\prod_{k=j+1}^n(k-j)\right)\right]^{-1} }{(\lambda+r)^{n-m}} \frac{1}{s+j(\lambda+r)} \\
   &= \frac{(n-1)!}{(m-1)!} \lambda^{n-m} \sum_{j=m}^n \frac{\left[ (-1)^{j-m} (j-m)! (n-j)! \right]^{-1} }{(\lambda+r)^{n-m}} \frac{1}{s+j(\lambda+r)} \\
   &= \lambda^{n-m} \sum_{j=m}^n \binom{n-1}{j-1} \binom{j-1}{m-1} \frac{(-1)^{j-m}}{(\lambda+r)^{n-m}} \frac{1}{s+j(\lambda+r)} .
\end{split}
\end{equation}
when $n>m$.  Inverse transforming with respect to $s$, we obtain
\begin{equation}
 f_{mm}(r,t) = e^{-m(\lambda+r)t}
\end{equation}
and
\begin{equation}
  f_{mn}(r,t) = \lambda^{n-m} \sum_{j=m}^n \binom{n-1}{j-1} \binom{j-1}{m-1} \frac{(-1)^{j-m}}{(\lambda+r)^{n-m}} e^{-j(\lambda+r)t}
\label{eq:fmnrtyule}
\end{equation}
when $n>m$. Again inverse transforming \eqref{eq:fmnrtyule}, this time with respect to the Laplace reward variable $r$, we find that for $m=n$,
\begin{equation}
 v_{mm}(x,t) = \delta(x-mt)e^{-m\lambda t} 
\end{equation}
which is a point mass at $x=mt$.  For $n>m$,
\begin{equation}
  v_{mn}(x,t) = \frac{\lambda^{n-m} e^{-\lambda x}}{(n-m-1)!} \sum_{j=m}^n \binom{n-1}{j-1} \binom{j-1}{m-1} (-1)^{j-m} (x-j t)^{n-m-1}  H(x-j t) .
\end{equation}
This completes the proof.
\end{proof}

\section{Proof of Theorem \ref{thm:generalyule}}

\label{app:generalyule}

\begin{proof}
Lemma \ref{lem:rewardodes} with arbitrary rewards $a_k$, $k=1,2,\ldots$, gives
\begin{equation} 
\begin{split}
  \od{f_{mn}(r,t)}{t} &= -(n\lambda + a_n r) f_{mn}(r,t) + (n-1)\lambda f_{m,n-1}(r,t) .
\end{split}
\end{equation}
To solve the system, apply the Laplace transform with respect to time $t$.  First note that the transform of $f_{mm}(r,t)$ is
\begin{equation}
 g_{mm}(r,s) = \frac{1}{s + m\lambda + a_m r}. 
\end{equation}
Transforming the $n$th equation, and recalling that $f_{mn}(r,0) = 0$ for $n>m$,
\begin{equation}
\begin{split}
 sg_{mn}(r,s)-f_{mn}(r,0) &= -(n\lambda + a_n r) g_{mn}(r,s) + (n-1)\lambda g_{m,n-1}(r,s) \\
 g_{mn}(r,s)(s + n\lambda + a_n r) &= (n-1)\lambda g_{m,n-1}(r,s) \\
 g_{mn}(r,s) &= \frac{(n-1)\lambda}{s+n\lambda + a_n r} g_{m,n-1}(r,s) \\
  &= \frac{ (n-1)!}{(m-1)!} \frac{\lambda^{n-m}}{\prod_{j=m+1}^n (s + j\lambda + a_j r)} g_{mm}(r,s)  \\
    &= \frac{(n-1)!}{(m-1)!} \frac{\lambda^{n-m}}{\prod_{j=m}^n (s + j\lambda + a_j r)} .
\end{split}
\end{equation}
We expand the denominator by partial fractions to find 
\begin{equation}
 g_{mn}(r,s)  = \lambda^{n-m} \frac{(n-1)!}{(m-1)!} \sum_{j=m}^n \frac{\prod_{k\neq j} \big[\lambda(k - j) + r(a_k - a_j)\big]^{-1}}{s + j\lambda + a_j r} .
\end{equation}
Transforming back to the time domain, we have, for $m=n$,
\begin{equation}
f_{mm}(r,t) = e^{-(m\lambda + a_mr)t} .
\end{equation}
When $n>m$,
\begin{equation}
 f_{mn}(r,t)  = \lambda^{n-m} \frac{(n-1)!}{(m-1)!} \sum_{j=m}^n \frac{e^{-(j\lambda + a_jr)t}}{\prod_{k\neq j} \big(\lambda(k - j) + r(a_k - a_j)\big)} .
\label{eq:fmn}
\end{equation}
This completes the proof.  
\end{proof}


\section{Analytic and numerical inversion}

\label{app:inversion}

Analytic inversion of \eqref{eq:generalyule} in Theorem \ref{thm:generalyule} is possible, but unfortunately depends on the structure of the tree topology in unexpected ways.  One convenient property of the rewards $\mathbf{a}=(a_1,\ldots,a_n)$ is that $a_i<a_{i+1}$ for all $1\leq i \leq n-1$, a fact apparent from \eqref{eq:omeara}.  Therefore $a_i\neq a_j$ for distinct $i$ and $j$.  When $m=n$, no speciation events have taken place, and we have 
\begin{equation}
  v_{mm}(x,t) = \delta(x-a_mt) e^{-m\lambda t}.
\end{equation}
For $n=m+1$, there is only one distinct topology, so 
\begin{equation}
\begin{split}
  v_{m,m+1}(x,t) &= \frac{m\lambda e^{-m\lambda t}}{a_{m+1}-a_m}\Bigg[ \exp\left(-\frac{\lambda(x-a_m t)}{a_{m+1}-a_m}\right)H(x-a_m t) \\ 
  & \qquad\qquad- e^{\lambda t} \exp\left(-\frac{\lambda(x-a_{m+1}t)}{a_{m+1}-a_m}\right)H(x-a_{m+1} t) \Bigg] .
\end{split}
\end{equation}
In general, when 
\begin{equation}
\frac{\ell-j}{a_\ell-a_j} - \frac{k-j}{a_k-a_j} \neq 0
\label{eq:denom}
\end{equation}
for any $l$, $k$, or $j$ in ${1,\ldots,n}$, then \eqref{eq:fmn} becomes
\begin{equation}
\begin{split}
 f_{mn}(r,t)  &= \lambda^{n-m} \frac{(n-1)!}{(m-1)!} \sum_{j=m}^n \frac{e^{-(j\lambda + a_jr)t}}{\prod_{k\neq j} (a_k-a_j)\left(\frac{\lambda(k - j)}{a_k-a_j} + r\right)} \\
 &= \lambda^{n-m} \frac{(n-1)!}{(m-1)!} \sum_{j=m}^n \frac{e^{-(j\lambda + a_jr)t}}{\prod_{k\neq j} (a_k-a_j)} \sum_{k\neq j} \frac{\prod_{\substack{\ell\neq k\\ \ell\neq j}} \left(\frac{\lambda(\ell - j)}{a_\ell-a_j} - \frac{\lambda(k-j)}{a_k-a_j}\right)^{-1}}{\frac{\lambda(k - j)}{a_k-a_j} + r} ,
\end{split}
\label{eq:fmn2}
\end{equation}
and so the full probability density for $n>m$ is 
\begin{equation}
 v_{mn}(x,t) = \lambda^{2} \frac{(n-1)!}{(m-1)!} \sum_{j=m}^n \frac{e^{-j\lambda t}H(x-a_jt)}{\prod_{k\neq j} (a_k-a_j)} \sum_{k\neq j} \frac{ \exp\left[-\frac{\lambda(k - j)}{a_k-a_j}(x-a_jt)\right] }{\prod_{\substack{\ell\neq k\\ \ell\neq j}} \left(\frac{\ell - j}{a_\ell-a_j} - \frac{k-j}{a_k-a_j}\right)} 
\label{eq:vmngeneral}
\end{equation}
However, the rewards $\mathbf{a}$ for many topologies do not satisfy \eqref{eq:denom}.  This can be seen in Figure \ref{fig:badrewards}, where $m=2$ and $n=4$.  Then we see that when $j=2$, $k=4$, and $\ell=3$ in \eqref{eq:vmngeneral},
\begin{equation}
\frac{4-2}{a_4-a_2} = \frac{4-2}{0.75-0.5} = 8
\end{equation}
and
\begin{equation}
\frac{3-2}{a_3-a_2} = \frac{3-2}{0.75-0.625} = 8 ,
\end{equation}
so the denominator in the second sum in \eqref{eq:vmngeneral} is zero.  Unfortunately this happens whenever there is symmetry in the tree so that more than one pair of taxa have the same time of shared ancestry.  Note also that \eqref{eq:vmngeneral} does not reduce to \eqref{eq:generalyule} in Theorem \ref{thm:generalyule} when $a_k=k$ since the denominator in the summand of \eqref{eq:vmngeneral} is zero.

\begin{figure}
\begin{center}


\includegraphics{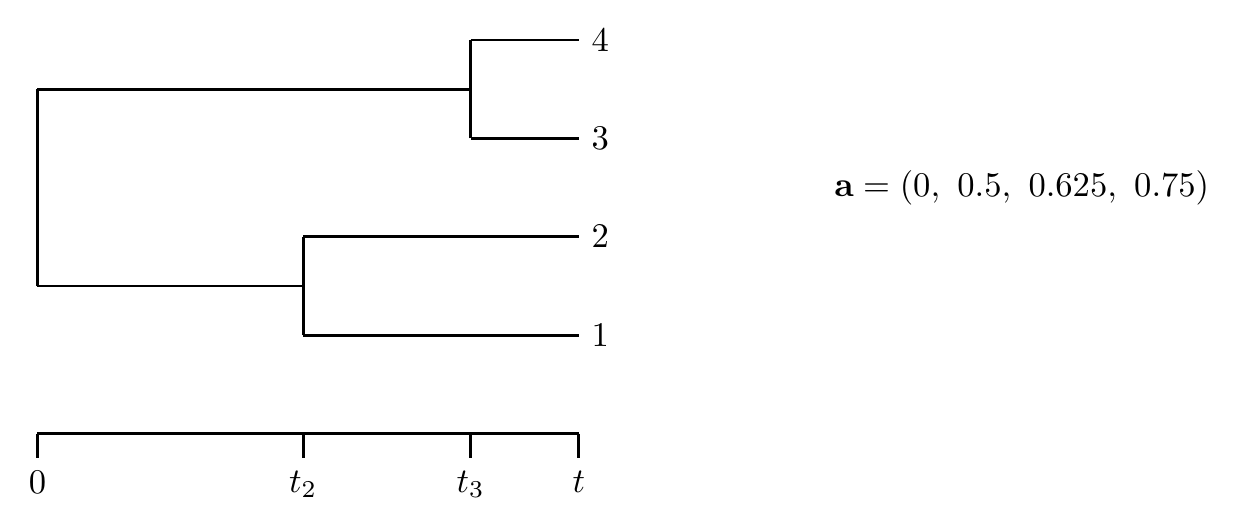}
\end{center}
\caption[Demonstration of a problematic reward vector]{Demonstration of a problematic reward vector $\mathbf{a}$ computed for the symmetric four-taxon tree.  In this case, the analytic inversion formula \eqref{eq:vmngeneral} cannot be applied, since the denominator in the sum becomes zero.  }
\label{fig:badrewards}
\end{figure}

Despite the difficulty in writing a general inversion to obtain $f_{mn}(x,t)$ for any topology, numerical inversion to arbitrary precision remains straightforward.  \citet{Abate1995Numerical} describe a numerical method for inverting the Laplace transform of probability densities by a discrete Riemann sum using the trapezoidal rule with step size $h$:
\begin{equation}
 v_{mn}(x,t) \approx \frac{e^{A/2}}{2x} \Re\left[f_{m,n}\left(\frac{A}{2x},t\right)\right] + \frac{e^{A/2}}{x} \sum_{k=1}^\infty (-1)^k \Re\left[f_{m,n}\left(\frac{A+2k\pi i}{2x},t\right)\right],  \\
\label{eq:quad}
\end{equation}
where we choose $A=20$.


\bibliographystyle{spbasic}
\bibliography{fcrawford}

\end{document}